\documentclass[copyright,creativecommons]{eptcs}
\providecommand{\event}{HPC'10}  
\providecommand{\firstpage}{1}
\usepackage{color}

\usepackage{amssymb}
\usepackage{amsfonts}
\usepackage{amsthm}
\usepackage{amsmath}
\usepackage{latexsym}
\usepackage{enumerate}

\usepackage{ifthen}
\usepackage{graphicx}
\usepackage{verbatim}
\usepackage{fix-cm}
\usepackage{hyperref}
\usepackage{stmaryrd}
\usepackage{xspace}
\usepackage{graphicx}
\graphicspath{{figures/}}
\usepackage{xy}
\xyoption{all}
\usepackage{color}
\usepackage{subfigure}

\usepackage{tikzfig}
\usetikzlibrary{shapes}
\def\bR{\begin{color}{red}}
\def\bB{\begin{color}{blue}}
\def\bM{\begin{color}{magenta}} 
\def\bC{\begin{color}{cyan}}
\def\bW{\begin{color}{white}}
\def\bBl{\begin{color}{black}}
\def\bG{\begin{color}{green}}
\def\bY{\begin{color}{yellow}}
\def\e{\end{color}}

\newcommand{\bit}{\begin{itemize}}
\newcommand{\eit}{\end{itemize}\par\noindent}
\newcommand{\ben}{\begin{enumerate}}
\newcommand{\een}{\end{enumerate}\par\noindent}
\newcommand{\beq}{\begin{equation}}
\newcommand{\eeq}{\end{equation}\par\noindent} 
\newcommand{\beqa}{\begin{eqnarray*}}
\newcommand{\eeqa}{\end{eqnarray*}\par\noindent}
\newcommand{\beqn}{\begin{eqnarray}}
\newcommand{\eeqn}{\end{eqnarray}\par\noindent}

\pgfdeclarelayer{edgelayer}
\pgfdeclarelayer{nodelayer}
\pgfsetlayers{edgelayer,nodelayer,main}

\tikzstyle{white dot}=[dot,fill=white]
\tikzstyle{none}=[inner sep=0mm]

\tikzstyle{dot}=[inner sep=0.5mm,fill=black,draw=black,shape=circle]
\tikzstyle{dotpic}=[baseline=-0.25em,shorten <=-0.1mm,shorten >=-0.1mm,scale=0.6]
\tikzstyle{small}=[inner sep=0.4mm]
\tikzstyle{dotpic inline}=[baseline=(current bounding box).south]
\tikzstyle{cdiag}=[baseline=(current bounding box).east,xscale=1.5,-latex]
\tikzstyle{every loop}=[]
\tikzstyle{rn}=[dot]
\tikzstyle{gn}=[dot]
\tikzstyle{bn}=[inner sep=0pt]
\tikzstyle{uploop}=[in=45,out=135,loop]
\tikzstyle{downloop}=[in=-45,out=-135,loop]
\tikzstyle{small dot}=[dot,inner sep=0.4mm]
\tikzstyle{small white dot}=[small dot,fill=white]
\tikzstyle{small gray dot}=[small dot,fill=gray!50]
\tikzstyle{greenbox}=[rectangle,fill=gray!30,draw=gray!50!black,minimum height=7mm,minimum width=7mm]
\tikzstyle{bluebox}=[rectangle,fill=white,draw=gray,minimum height=7mm,minimum width=7mm]
\tikzstyle{cnot}=[fill=white,shape=circle,inner sep=-1.4pt]
\tikzstyle{pt}=[regular polygon,regular polygon sides=3,draw=black,scale=0.75,inner sep=-0.5pt]
\tikzstyle{copt}=[pt,regular polygon rotate=180]
\tikzstyle{tick}=[sloped,rotate=90,font=\small\bf,xshift=0.07mm]
\tikzstyle{pauli z}=[sloped,rotate=90,font=\scriptsize\bf]
\tikzstyle{white dot}=[dot,fill=white]
\tikzstyle{gray dot}=[dot,fill=gray!50]
\tikzstyle{gs dot}=[dot,fill=gray]
\tikzstyle{small dotpic}=[dotpic,scale=0.6]
\tikzstyle{mux}=[rectangle,draw=black,scale=0.5,minimum width=1.8cm,minimum height=1cm]

\tikzstyle{square box}=[rectangle,fill=white,draw=black,minimum height=6mm,minimum width=6mm]
\tikzstyle{square gray box}=[rectangle,fill=gray!30,draw=black,minimum height=6mm,minimum width=6mm]
\tikzstyle{dashed box}=[draw=black,dashed,minimum height=12mm,minimum width=12mm,fill=gray!20]
\tikzstyle{box vertex}=[rectangle,draw=black]

\newtheorem{theorem}{Theorem}[section]
\newtheorem{proposition}[theorem]{Proposition}
\newtheorem{lemma}[theorem]{Lemma}
\newtheorem{corollary}[theorem]{Corollary}
\theoremstyle{definition}\newtheorem{example}[theorem]{Example}
\theoremstyle{definition}
\theoremstyle{definition}\newtheorem{definition}[theorem]{Definition}
\theoremstyle{definition}
\theoremstyle{definition}\newtheorem{remark}[theorem]{Remark}
\theoremstyle{definition}
\theoremstyle{definition}
\theoremstyle{definition}
\theoremstyle{definition}

\newcommand{\ket}[1]{|#1\rangle}
\newcommand{\bra}[1]{\langle#1|}

\newcommand{\ketGHZ}{\ket{\textit{GHZ}\,}}
\newcommand{\ketW}{\ket{\textit{W\,}}}

\newcommand{\dotunit}[1]{%
\begin{tikzpicture}[dotpic,yshift=-1mm]
\node [#1] (0) at (0,0.25) {}; 
\draw (0)--(0,-0.15);
\end{tikzpicture}\,}

\newcommand{\dotcounit}[1]{%
\begin{tikzpicture}[dotpic,yshift=1mm]
\node [#1] (0) at (0,-0.25) {}; 
\draw (0)--(0,0.15);
\end{tikzpicture}\,}

\newcommand{\dotcap}[1]{%
\begin{tikzpicture}[dotpic,yshift=-1mm]
	\node [#1] (a) at (0,0.25) {};
	\draw [bend right] (a) to (-0.4,-0.15) (0.4,-0.15) to (a);
\end{tikzpicture}}

\newcommand{\dotcup}[1]{%
\begin{tikzpicture}[dotpic,yshift=1mm]
	\node [#1] (a) at (0,-0.25) {};
	\draw [bend right] (-0.4,0.15) to (a) (a) to (0.4,0.15);
\end{tikzpicture}}

\newcommand{\dotmult}[1]{%
\begin{tikzpicture}[dotpic]
	\node [#1] (a) {};
	\draw (a) -- (-90:0.35);
	\draw (a) -- (45:0.4);
	\draw (a) -- (135:0.4);
\end{tikzpicture}}

\newcommand{\dotcomult}[1]{%
\begin{tikzpicture}[dotpic]
	\node [#1] (a) {};
	\draw (a) -- (90:0.35);
	\draw (a) -- (-45:0.4);
	\draw (a) -- (-135:0.4);
\end{tikzpicture}}

\newcommand{\dottickunit}[1]{%
\begin{tikzpicture}[dotpic,yshift=-1mm]
\node [#1] (0) at (0,0.25) {}; 
\draw (0)-- node[tick]{-} (0,-0.15);
\end{tikzpicture}}

\newcommand{\dotcrossunit}[1]{%
\begin{tikzpicture}[dotpic,yshift=-1mm]
\node [#1] (0) at (0,0.25) {}; 
\draw (0)-- node[pauli z]{$\times$} (0,-0.15);
\end{tikzpicture}}

\newcommand{\tick}{%
\begin{tikzpicture}[dotpic]
	\node [style=none] (0) at (0,0.25) {};
	\node [style=none] (1) at (0,-0.35) {};
	\draw (0) -- node[tick]{-} (1);
\end{tikzpicture}}

\newcommand{\cross}{%
\begin{tikzpicture}[dotpic]
	\node [style=none] (0) at (0,0.35) {};
	\node [style=none] (1) at (0,-0.35) {};
	\draw (0) -- node[pauli z]{$\times$} (1);
\end{tikzpicture}}

\newcommand{\point}[1]{%
\begin{tikzpicture}[dotpic]
	\begin{pgfonlayer}{nodelayer}
		\node [style=white dot] (0) at (0, 0.3) {$#1$};
		\node [style=none] (1) at (0, -0.6) {};
	\end{pgfonlayer}
	\begin{pgfonlayer}{edgelayer}
		\draw (0) to (1.center);
	\end{pgfonlayer}
\end{tikzpicture}}

\newcommand{\lolli}{%
\begin{tikzpicture}[dotpic]
	\path [use as bounding box] (-0.25,-0.25) rectangle (0.25,0.5);
	\node [style=small dot] (0) at (0, 0.15) {};
	\node [style=none] (1) at (0, -0.25) {};
	\draw  (0) to (1.center);
	\draw [out=45, looseness=2.00, in=135, loop] (0) to ();
\end{tikzpicture}}

\newcommand{\cololli}{%
\begin{tikzpicture}[dotpic]
	\path [use as bounding box] (-0.25,-0.5) rectangle (0.25,0.5);
	\node [style=none] (0) at (0, 0.5) {};
	\node [style=small dot] (1) at (0, 0) {};
	\draw [out=-45, looseness=2.00, in=-135, loop] (1) to ();
	\draw  (1) to (0.center);
\end{tikzpicture}}

\newcommand{\unit}{\dotunit{small dot}}
\newcommand{\counit}{\dotcounit{small dot}}
\newcommand{\mult}{\dotmult{small dot}}
\newcommand{\comult}{\dotcomult{small dot}}
\newcommand{\tickunit}{\dottickunit{small dot}}

\newcommand{\blackcup}{\dotcup{small dot}}
\newcommand{\blackcap}{\dotcap{small dot}}

\newcommand{\whiteunit}{\dotunit{small white dot}}
\newcommand{\whitecounit}{\dotcounit{small white dot}}
\newcommand{\whitemult}{\dotmult{small white dot}}
\newcommand{\whitecomult}{\dotcomult{small white dot}}

\newcommand{\whitecrossunit}{\dotcrossunit{small white dot}}

\newcommand{\circl}{\begin{tikzpicture}[dotpic]
	\node [circle,draw=black,inner sep=1pt] {\footnotesize\sf\phantom{$-$}};
\end{tikzpicture}}

\newcommand{\tickpsi}{
\raisebox{-1mm}{\begin{tikzpicture}[scale=0.6]
	\begin{pgfonlayer}{nodelayer}
		\node [style=none] (0) at (0,0.25) {};
		\node [style=none,font=\small\bf,xshift=0.07mm] (2) at (0,0) {-};
		\node [style=none] (1) at (0,-0.25) {};
	\end{pgfonlayer}
	\begin{pgfonlayer}{edgelayer}
		\draw (0) to (2.center) (2.center) to (1);
	\end{pgfonlayer}
\end{tikzpicture}}\circ\psi}

\newcommand{\mdots}{\overbrace{
\begin{tikzpicture}
	\begin{pgfonlayer}{nodelayer}
		\node [style=white dot] (0) at (0.75, 1.5) {};
		\node [style=none] (1) at (1.25, 1.5) {$\ldots$};
		\node [style=white dot] (2) at (1.75, 1.5) {};
		\node [style=dot] (3) at (1.25, 1) {};
	\end{pgfonlayer}
	\begin{pgfonlayer}{edgelayer}
		\draw (2) to (3);
		\draw (0) to (3);
	\end{pgfonlayer}
\end{tikzpicture}}^m}

\def\II{{\rm I}}

\title{The GHZ/W-calculus contains rational arithmetic}

\author{
Bob Coecke\thanks{Supported by EPSRC Advanced Research Fellowship EP/D072786/1, Office of Naval Research Grant N00014-09-1-0248 and EC FP6 STREP QICS.
Some of this work was performed during a visit at IQOQI Vienna.},
Aleks Kissinger\thanks{Supported by a Clarendon scholarship},
Alex Merry\thanks{Supported by a DTA scholarship}
\institute{Oxford University Computing Laboratory, Quantum Group\\
Wolfson Building, Parks Road,
Oxford OX1 3QD, UK}
\email{coecke/alek/alemer@comlab.ox.ac.uk}
\and
Shibdas Roy\thanks{Funded by DST, Govt. of India}
\institute{Center for~Quant.~Inf.~and Quant.~Comp.\\
Department of Physics,
IISc, Bangalore}
\email{roy\underline{\ }shibdas@yahoo.co.in}
}
\def\authorrunning{B. Coecke, A. Kissinger, A. Merry \& S. Roy}
\def\titlerunning{The GHZ/W-calculus contains rational arithmetic}

\begin{document}
\maketitle

\begin{abstract}
Graphical calculi for representing interacting quantum systems serve
a number of purposes: compositionally, intuitive graphical
reasoning, and a logical underpinning for automation.  The power of
these calculi stems from the fact that they embody generalized
symmetries of the structure of quantum operations, which, for
example, stretch well beyond the Choi-Jamiolkowski isomorphism.  One
such calculus takes the GHZ and W states as its basic generators.
Here we show that this  language allows one to encode standard
rational calculus, with the GHZ state as multiplication, the W state
as addition, the Pauli X gate as multiplicative inversion, and the
Pauli Z gate as additive inversion.
\end{abstract}

\section{Introduction}

\em Categorical quantum mechanics \em \cite{AC} aims to recast quantum
mechanical notions in terms of symmetric monoidal categories with
additional structure.  One layer of extra structure, compactness
\cite{KellyLaplaza}, encompasses the well-known Choi-Jamiolkowski
isomorphism.  Compactness is itself subsumed by the much richer
commutative Frobenius algebra structure \cite{CarboniWalters}, which
governs classical data, observables, and certain tripartite states
\cite{CPav, CD, CES2, CK}.  In this symmetric monoidal form,
quantum mechanics enjoys:
\bit
\item an \em operational interpretation \em by making sequential and
      parallel composition of systems and processes the basic
      connectives of the language \cite{ContPhys};
\item an intuitive \em diagrammatic calculus \em \cite{ContPhys} via the
      Penrose-Joyal-Street diagrammatic calculus for symmetric monoidal
      categories \cite{Penrose, JS}, augmented with Kelly and Laplaza's
      coherence result for compact categories,  and Lack's work on
      distributive laws \cite{Lack};
\item a \em logical underpinning \em \cite{RossThesis} via the closed
      structure resulting from compactness.
\eit
The last allows the application of automated reasoning techniques to
quantum mechanics \cite{DD, quanto, DK}.  A prototype software
implementation, {\tt quantomatic}, already exists and is jointly
developed in Edinburgh and Oxford.

Categorical quantum mechanics has meanwhile been successful in
solving problems in quantum information \cite{DP} and quantum
foundations \cite{CES2}, where other methods and structures failed
to be adequate. Key to these results is the description of \em
interacting basis structures \em in \cite{CD}. The language of that
paper consists of a pair of abstract bases or \em basis
structures\em, which are, again in abstract terms, mutually
unbiased, and an abstract generalisation of phases relative to
bases.  This formalism has been implemented in {\tt quantomatic},
and is expressive enough to universally model any linear map
$f:\mathbb{Q}^{\otimes n}\to \mathbb{Q}^{\otimes m}$, where
$\mathbb{Q}=\mathbb{C}^2$.  On the other hand, if we restrict the
language to the two basis structures only it becomes very poor,
describing no more than 2 qubit states.

This brings us to the subject of this paper.  In \cite{CK} two of
the authors introduced pairs of interacting commutative Frobenius
algebras that do not model bases, but the tripartite GHZ and W
states \cite{DVC}.  Both these states can indeed be endowed with the
structure of a commutative Frobenius algebra, yielding a \em GHZ
structure \em and  a \em W structure \em as we recall in Section
\ref{sec:ghz-w}. The main point of this paper is that the language
consisting of the GHZ structure (which is essentially the same as a
basis structure) and the W structure is already rich enough to
encode rational arithmetic, with the exception of additive inverses.
Now an infinite number of qubit states can be described,
corresponding to the rational numbers of the arithmetic system. We
demonstrate this in Section \ref{sec:arithmetic}. In Section
\ref{sec:additive-inverses} we extend the GHZ/W-calculus with one
basic graphical element which then allows additive inverses to be
captured.  Section \ref{sec:automation} addresses the issue of how
to implement the calculus within the {\tt quantomatic} software.

We assume that the reader is familiar with the diagrammatic calculus
for symmetric monoidal categories \cite{JS,SelingerSurvey}, which is
also reviewed in \cite{CK}. We also assume that the reader is
familiar with the (very) basics of finite dimensional Hilbert spaces
and Dirac notation as used in quantum computing.

\section{Frobenius Algebras and the GHZ/W-calculus}\label{sec:ghz-w}

Fix a symmetric monoidal category $({\bf V},\otimes,I,\sigma)$. Throughout this paper, we shall
define morphisms in $\bf V$ using the graphical notation defined in \cite{SelingerSurvey}. In this
notation, `wires' correspond to objects and vertices, and `boxes' correspond to morphisms. We shall
express composition vertically, from top to bottom, and the monoidal product as (horizontal)
juxtaposition of graphs. When wires are not labeled, they are assumed to represent a fixed object,
$Q$.

\begin{example}
A canonical example throughout will be ${\bf FHilb}$, the category of finite-dimensional Hilbert
spaces and linear maps. In this case, $\otimes$ is the usual tensor product, $\sigma$ the swap map
$v \otimes w \mapsto w \otimes v$, $I := \mathbb C$ and $Q := \mathbb C^2$, the space of qubits. We
shall also refer the ``projective'' category of finite-dimensional Hilbert spaces, ${\bf FHilb}_p$,
whose objects are the same as ${\bf FHilb}$ and whose arrows are linear maps, taken to be
equivalent iff they differ only by a non-zero scalar.
\end{example}

\subsection{Commutative Frobenius Algebras}

A \em commutative Frobenius algebra \em (CFA) consists of an internal  commutative monoid $(Q, \mult\ , \unit)$ and an internal  cocommutative comonoid $(Q, \comult\ , \counit)$ that interact via the Frobenius law:
 \[

\]

One can show that any connected graph consisting only of $\mult, \unit, \comult, \counit, \sigma$
and $1_Q$ depends only upon the number of inputs, outputs, and loops. As such, it can be reduced to
a canonical normal form:
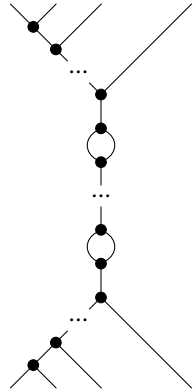
\[
%
\]

In any connected graph, loops are counted as the total number of edges that can be removed without
disconnecting the graph. We shall use `spider' notation to represent graphs of Frobenius algebras
using vertices of any arity. We express any connected graph as above with $m$ inputs, $n$ outputs,
and no loops as a single vertex of the same colour:
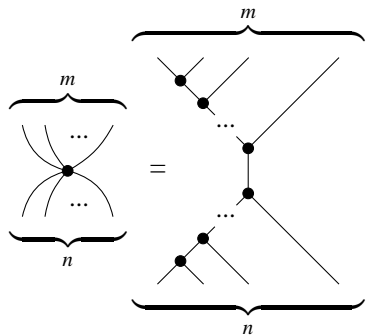
\[
%
\]

We give two of these graphs special names.  The \emph{cup} is defined as $\blackcup$ and the \emph{cap} is defined as $\blackcap$.  These induce a compact structure, since
\[
%
\]

\subsection{Phases}

\begin{definition}{\rm\cite{CD}}\label{Phasedef}
Given a  CFA on an object $A$, a morphism $f:A\to A$ is a \em phase \em if we have
\beq\label{eq:phases1}
\raisebox{-8mm}{%
}
\eeq
\end{definition}

Equivalently, phases can be described as module endomorphisms, where $\mult$ is considered as a left (or right) module over itself.

\begin{proposition}\label{prop:phases-states}
A phase $f:A\to A$ 
can be equivalently defined as a morphism of the form:
\beq\label{eq:phases2}
\raisebox{-8mm}{%
}
\eeq
for some element $\psi:\II\to A$.
\end{proposition}
\begin{proof}
Given eqs.~(\ref{eq:phases1}) we have
\[
\raisebox{-8mm}{%
}
\]
where we used unitality of the CFA, and conversely, given eq.~(\ref{eq:phases2}),  eqs.~(\ref{eq:phases1}) straightforwardly follow by associativity and commutativity of the CFA.
\end{proof}

\begin{proposition}\label{prop:invphase}
The inverse of a phase is a phase.
\end{proposition}
\begin{proof}
Setting $
\raisebox{-4mm}{
%
\]
\end{proof}

\subsection{GHZ/W calculus}

In this paper, we are concerned not only with general CFAs, but two specific cases, depending on the behaviour of the loops. We refer to $\mult \circ \comult$ as the \emph{loop map} of a CFA.

\begin{definition}{\rm\cite{CK}}
    A \emph{GHZ-structure} is a special commutative Frobenius algebra; that is, a commutative Frobenius algebra where the loop map is equal to the identity:
\[
%
\]
\end{definition}

These GHZ-structures have also been referred to as \emph{basis structures}, for example in \cite{CES2}, because of their strong connection to bases in finite-dimensional vector spaces. See Theorem \ref{basisthm} below.

\begin{definition}{\rm\cite{CK}}
 A \emph{W-structure} is an anti-special commutative Frobenius algebra. This is commutative Frobenius algebra whose loop map obeys the following equation:
\begin{equation}\label{eq:antispec}
    \circl\ %
\end{equation}
where we use the following short-hand notation:
\[
%
\]
\end{definition}

This distinction essentially comes down to whether the loop map is singular or invertible.

\begin{lemma}\label{lem:iso-scfa}
    If the loop map of a CFA is an isomorphism, the CFA can be made special via a phase.
\end{lemma}
\begin{proof}
Consider a CFA $(\mult, \unit, \comult, \counit)$. Since its loop map is a phase, by Proposition \ref{prop:invphase} so is the inverse of the loop map, which we denote $f$.  Then
$(\,\mult\, , \, \unit\,,
\raisebox{-4.5mm}{
%
}
,\,\cololli\,)$ is easily seen to be a special CFA.
\end{proof}

\begin{lemma}[Herrmann \cite{Hermann}]
If the loop of a CFA is disconnected, i.e.~factors over the tensor unit, then it obeys eq.~{\rm(\ref{eq:antispec})}, that is the CFA is necessarily anti-special.
\end{lemma}

The following is an example of a GHZ-structure in ${\bf FHilb}$:
    \begin{equation}\label{GHZ-SCFA}
        \begin{split}
            \whitemult & = \ket{0}\bra{00} + \ket{1}\bra{11} \qquad\qquad
            \whiteunit = \sqrt{2}\, \ket{+} := \ket{0}+\ket{1} \\
            \whitecomult & = \ket{00}\bra{0} + \ket{11}\bra{1} \qquad\qquad
            \whitecounit = \sqrt{2} \bra{+} := \bra{0}+\bra{1}
        \end{split}\vspace{-1.5mm}
    \end{equation}
and we also have an example  of a W-structure in ${\bf FHilb}$:
  \begin{equation}\label{W-ACFA}
        \begin{split}
            \mult & = \ket{1}\bra{11} + \ket{0}\bra{01} + \ket{0}\bra{10}
            \qquad\qquad\qquad
            \unit = \ket 1\qquad\qquad \\
            \comult & = \ket{00}\bra{0} + \ket{01}\bra{1} + \ket{10}\bra{1}
            \qquad\qquad\qquad
            \counit = \bra 0\qquad\qquad
        \end{split}
    \end{equation}
Note that the cups for these CFAs do not coincide:
\[
\begin{tikzpicture}[dotpic]
    \begin{pgfonlayer}{nodelayer}
        \node [style=white dot] (0) at (1.5, 0.75) {};
        \node [style=white dot] (1) at (-1.5, 0.25) {};
        \node [style=white dot] (2) at (1.5, 0.25) {};
        \node [style=none] (3) at (0, 0) {$:=$};
        \node [style=none] (4) at (-2, -0.25) {};
        \node [style=none] (5) at (-1, -0.25) {};
        \node [style=none] (6) at (1, -0.25) {};
        \node [style=none] (7) at (2, -0.25) {};
    \end{pgfonlayer}
    \begin{pgfonlayer}{edgelayer}
        \draw (0) to (2);
        \draw (7.center) to (2);
        \draw[bend right=45] (5.center) to (1);
        \draw (6.center) to (2);
        \draw[bend left=45] (4.center) to (1);
    \end{pgfonlayer}
\end{tikzpicture}\ = \ket{00}+\ket{11}
\qquad\qquad
\begin{tikzpicture}[dotpic]
    \begin{pgfonlayer}{nodelayer}
        \node [style=dot] (0) at (1.5, 0.75) {};
        \node [style=dot] (1) at (-1.5, 0.25) {};
        \node [style=dot] (2) at (1.5, 0.25) {};
        \node [style=none] (3) at (0, 0) {$:=$};
        \node [style=none] (4) at (-2, -0.25) {};
        \node [style=none] (5) at (-1, -0.25) {};
        \node [style=none] (6) at (1, -0.25) {};
        \node [style=none] (7) at (2, -0.25) {};
    \end{pgfonlayer}
    \begin{pgfonlayer}{edgelayer}
        \draw (0) to (2);
        \draw (7.center) to (2);
        \draw[bend right=45] (5.center) to (1);
        \draw (6.center) to (2);
        \draw[bend left=45] (4.center) to (1);
    \end{pgfonlayer}
\end{tikzpicture}\ = \ket{01}+\ket{10}
\]

However, the composition of a cap from one CFA with a cup from the other yields the Pauli X, or `NOT', gate:
\[
\begin{tikzpicture}[dotpic,scale=0.5]
    \node [bn] (0) at (0,1.5) {};
    \node [bn] (1) at (0,-1.5) {};
    \draw (0)-- node[tick]{-} (1);
\end{tikzpicture}
:=\
\begin{tikzpicture}[dotpic,yshift=-5mm,scale=0.5]
    \node [bn] (b0) at (-1,2) {};
    \node [dot] (0) at (0,0) {};
    \node [white dot] (1) at (1.5,1) {};
    \node [bn] (b1) at (2.5,-1) {};
    \draw (b0) to [out=-90,in=180] (0) (0) to [out=0,in=180] (1) (1) to [out=0,in=90] (b1);
\end{tikzpicture} \ =\
\begin{tikzpicture}[dotpic,yshift=-5mm,scale=0.5]
    \node [bn] (b0) at (-1,2) {};
    \node [white dot] (0) at (0,0) {};
    \node [dot] (1) at (1.5,1) {};
    \node [bn] (b1) at (2.5,-1) {};
    \draw (b0) to [out=-90,in=180] (0) (0) to [out=0,in=180] (1) (1) to [out=0,in=90] (b1);
\end{tikzpicture}
\ =\
\left(\begin{array}{cc}
0 & 1 \\
1 & 0
\end{array}\right)
\]
These CFAs respectively induce the following tripartite states:
\[
\begin{tikzpicture}[dotpic]
    \begin{pgfonlayer}{nodelayer}
        \node [style=white dot] (0) at (0, 0.75) {};
        \node [style=white dot] (1) at (0, 0) {};
        \node [style=white dot] (2) at (-0.5, -0.5) {};
        \node [style=none] (3) at (-1, -1) {};
        \node [style=none] (4) at (0, -1) {};
        \node [style=none] (5) at (1, -1) {};
    \end{pgfonlayer}
    \begin{pgfonlayer}{edgelayer}
        \draw (4.center) to (2);
        \draw (3.center) to (2);
        \draw (2) to (1);
        \draw (0) to (1);
        \draw (5.center) to (1);
    \end{pgfonlayer}
\end{tikzpicture}
\ = \ket{000}+\ket{111} = \ketGHZ
\qquad
\begin{tikzpicture}[dotpic]
    \begin{pgfonlayer}{nodelayer}
        \node [style=dot] (0) at (0, 0.75) {};
        \node [style=dot] (1) at (0, 0) {};
        \node [style=dot] (2) at (-0.5, -0.5) {};
        \node [style=none] (3) at (-1, -1) {};
        \node [style=none] (4) at (0, -1) {};
        \node [style=none] (5) at (1, -1) {};
    \end{pgfonlayer}
    \begin{pgfonlayer}{edgelayer}
        \draw (4.center) to (2);
        \draw (3.center) to (2);
        \draw (2) to (1);
        \draw (0) to (1);
        \draw (5.center) to (1);
    \end{pgfonlayer}
\end{tikzpicture}
\ = \ket{100}+\ket{010}+\ket{001} = \ketW
\]
As the name suggests, the associated tripartite state of the above GHZ-structure is a GHZ state,
and that of the W-structure is a W state. Furthermore, Theorem \ref{GHZ/Wthm} asserts that for
qubits, the associated tripartite state of \emph{any} GHZ-structure (resp. W-structure) is a GHZ
state (resp. W state), up to local operations.

\begin{theorem}{\rm\cite{CK}}\label{GHZ/Wthm}
For any special (respectively~anti-special) CFA on a qubit in ${\bf FHilb}$, the induced tripartite
state is SLOCC-equivalent to $\ketGHZ$ (respectively~$\ketW$).  Furthermore, any tripartite state
$\ket\Psi$ either induces a special or anti-special CFA-structure, depending on whether it is
SLOCC-equivalent to $\ketGHZ$ or to $\ketW$.
\end{theorem}

Theorem \ref{basisthm} justifies the alternative name \em basis structure \em for GHZ-structures.

\begin{theorem}{\rm\cite{Aguiar}}\label{basisthm}
Special commutative Frobenius algebras on a finite-dimensional Hilbert space ${\cal H}$ are in 1-to-1 correspondence with (possibly non-orthogonal) bases for ${\cal H}$.
\end{theorem}

For any special CFA, phases are matrices that are diagonal in the corresponding basis. The
corresponding $|\psi\rangle$ (as in proposition \ref{prop:phases-states}) lies on the equator of
the Bloch sphere, justifying the name `phases'.

We can also consider interactions between a GHZ-structure and a W-structure.

\begin{definition}{\rm\cite{CK}}
A  GHZ- and a W-structure form a \em GHZ/W-pair \em if the following equations hold:
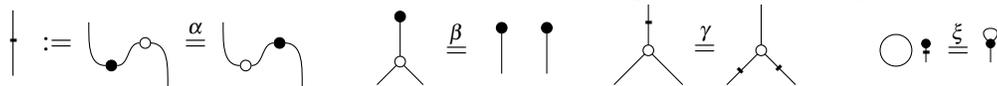
\begin{center}
\raisebox{4mm}{
\begin{tikzpicture}[dotpic,scale=0.5]
    \node [bn] (0) at (0,1.5) {};
    \node [bn] (1) at (0,-1.5) {};
    \draw (0)-- node[tick]{-} (1);
\end{tikzpicture}
 \ $:=$\ \
\begin{tikzpicture}[dotpic,yshift=-5mm,scale=0.5]
    \node [bn] (b0) at (-1,2) {};
    \node [dot] (0) at (0,0) {};
    \node [white dot] (1) at (1.5,1) {};
    \node [bn] (b1) at (2.5,-1) {};
    \draw (b0) to [out=-90,in=180] (0) (0) to [out=0,in=180] (1) (1) to [out=0,in=90] (b1);
\end{tikzpicture}
 \ $\stackrel{\alpha}{=}$\  \
\begin{tikzpicture}[dotpic,yshift=-5mm,scale=0.5]
    \node [bn] (b0) at (-1,2) {};
    \node [white dot] (0) at (0,0) {};
    \node [dot] (1) at (1.5,1) {};
    \node [bn] (b1) at (2.5,-1) {};
    \draw (b0) to [out=-90,in=180] (0) (0) to [out=0,in=180] (1) (1) to [out=0,in=90] (b1);
\end{tikzpicture}
}
\qquad
\begin{tikzpicture}[dotpic]
    \begin{pgfonlayer}{nodelayer}
        \node [style=none] (0) at (2.5, 1.5) {};
        \node [style=none] (1) at (5, 1.5) {};
        \node [style=dot] (2) at (-3, 1.25) {};
        \node [style=dot] (3) at (-0.75, 1) {};
        \node [style=dot] (4) at (0.25, 1) {};
        \node [style=none] (5) at (-1.75, 0.75) {$\stackrel{\beta}{=}$};
        \node [style=none] (6) at (3.75, 0.75) {$\stackrel{\gamma}{=}$};
        \node [style=white dot] (7) at (2.5, 0.5) {};
        \node [style=white dot] (8) at (5, 0.5) {};
        \node [style=white dot] (9) at (-3, 0.25) {};
        \node [style=none] (10) at (-0.75, 0) {};
        \node [style=none] (11) at (0.25, 0) {};
        \node [style=none] (12) at (-3.5, -0.25) {};
        \node [style=none] (13) at (-2.5, -0.25) {};
        \node [style=none] (14) at (1.75, -0.25) {};
        \node [style=none] (15) at (3.25, -0.25) {};
        \node [style=none] (16) at (4.25, -0.25) {};
        \node [style=none] (17) at (5.75, -0.25) {};
    \end{pgfonlayer}
    \begin{pgfonlayer}{edgelayer}
        \draw (7) to (15.center);
        \draw (2) to (9);
        \draw (12.center) to (9);
        \draw (9) to (13.center);
        \draw (8) to node[tick]{-} (17.center);
        \draw (16.center) to node[tick]{-} (8);
        \draw (0.center) to node[tick]{-} (7);
        \draw (3) to (10.center);
        \draw (1.center) to (8);
        \draw (14.center) to (7);
        \draw (4) to (11.center);
    \end{pgfonlayer}
\end{tikzpicture}
\qquad\ \ \raisebox{3mm}{\ensuremath{\circl \tickunit \stackrel{\xi}{=}\, \lolli}}
\end{center}
\end{definition}

By eqs.~($\beta$, $\gamma$) we also have:
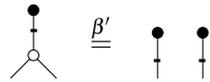
\vspace{-5mm}
\begin{center}
\begin{tikzpicture}[dotpic]
    \begin{pgfonlayer}{nodelayer}
        \node [style=dot] (0) at (-1.5, -0.75) {};
        \node [style=none] (1) at (0, -1.25) {$\stackrel{\beta'}{=}$};
        \node [style=dot] (2) at (1.25, -1.25) {};
        \node [style=dot] (3) at (2.25, -1.25) {};
        \node [style=white dot] (4) at (-1.5, -1.75) {};
        \node [style=none] (5) at (-2, -2.25) {};
        \node [style=none] (6) at (-1, -2.25) {};
        \node [style=none] (7) at (1.25, -2.25) {};
        \node [style=none] (8) at (2.25, -2.25) {};
    \end{pgfonlayer}
    \begin{pgfonlayer}{edgelayer}
        \draw (0) to node[tick]{-} (4);
        \draw (3) to node[tick]{-} (8.center);
        \draw (5.center) to (4);
        \draw (2) to node[tick]{-} (7.center);
        \draw (4) to (6.center);
    \end{pgfonlayer}
\end{tikzpicture}
\end{center}

\subsection{Plugging}

Since we are often concerned with objects in a monoidal category that are finitary in nature, we can deduce many new identities using a technique we call \emph{plugging}.

\begin{definition}\label{def:plugging-set}
    A set of points $\{ \psi_i : I \rightarrow Q \}$ form a \emph{plugging set} for $Q$ if they suffice to distinguish maps from $Q$. That is, for all objects $A$ and maps $f,g : Q \rightarrow A$,
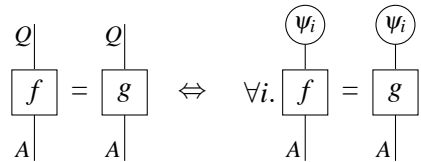
    \[\begin{tikzpicture}[dotpic]
        \begin{pgfonlayer}{nodelayer}
            \node [style=none] (0) at (-3.5, 1.5) {};
            \node [style=none] (1) at (-1.5, 1.5) {};
            \node [style=white dot, font=\footnotesize] (2) at (2.5, 1.5) {$\psi_i$};
            \node [style=white dot, font=\footnotesize] (3) at (4.5, 1.5) {$\psi_i$};
            \node [style=none, font=\footnotesize] (4) at (-3.75, 1.25) {$Q$};
            \node [style=none, font=\footnotesize] (5) at (-1.75, 1.25) {$Q$};
            \node [style=square box] (6) at (-3.5, 0) {$f$};
            \node [style=none] (7) at (-2.5, 0) {=};
            \node [style=square box] (8) at (-1.5, 0) {$g$};
            \node [style=none] (9) at (0, 0) {$\Leftrightarrow$};
            \node [style=none] (10) at (1.5, 0) {$\forall i .$};
            \node [style=square box] (11) at (2.5, 0) {$f$};
            \node [style=none] (12) at (3.5, 0) {=};
            \node [style=square box] (13) at (4.5, 0) {$g$};
            \node [style=none, font=\footnotesize] (14) at (-3.75, -1.25) {$A$};
            \node [style=none, font=\footnotesize] (15) at (-1.75, -1.25) {$A$};
            \node [style=none, font=\footnotesize] (16) at (2.25, -1.25) {$A$};
            \node [style=none, font=\footnotesize] (17) at (4.25, -1.25) {$A$};
            \node [style=none] (18) at (-3.5, -1.5) {};
            \node [style=none] (19) at (-1.5, -1.5) {};
            \node [style=none] (20) at (2.5, -1.5) {};
            \node [style=none] (21) at (4.5, -1.5) {};
        \end{pgfonlayer}
        \begin{pgfonlayer}{edgelayer}
            \draw (0.center) to (18.center);
            \draw (1.center) to (19.center);
            \draw (3) to (21.center);
            \draw (2) to (20.center);
        \end{pgfonlayer}
    \end{tikzpicture}\]
\end{definition}

When we prove a graphical identity by showing two maps are not distinguished by a plugging set, we
call this `proof by plugging.' Also note that  we can extend such proofs to maps of the form $f : Q
\otimes A \rightarrow B$ or $f' : A \rightarrow Q \otimes B$ by using the Frobenius caps and cups
when $Q$ has a CFA $(\mult, \unit, \comult, \counit)$ and $A$ a CFA $(\whitemult, \whiteunit,
\whitecomult, \whitecounit)$,
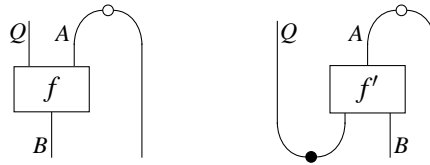
\[
\begin{tikzpicture}[dotpic]
    \begin{pgfonlayer}{nodelayer}
        \node [style=white dot] (0) at (-2.25, 1.75) {};
        \node [style=white dot] (1) at (4.25, 1.75) {};
        \node [style=none] (2) at (-4, 1.5) {};
        \node [style=none] (3) at (1.5, 1.5) {};
        \node [style=none, font=\footnotesize] (4) at (-4.25, 1.25) {$Q$};
        \node [style=none, font=\footnotesize] (5) at (-3.25, 1.25) {$A$};
        \node [style=none, font=\footnotesize] (6) at (1.75, 1.25) {$Q$};
        \node [style=none, font=\footnotesize] (7) at (3.25, 1.25) {$A$};
        \node [style=none] (8) at (-3, 1) {};
        \node [style=none] (9) at (-1.5, 1) {};
        \node [style=none] (10) at (3.5, 1) {};
        \node [style=none] (11) at (5, 1) {};
        \node [style=none] (12) at (-4, 0.5) {};
        \node [style=none] (13) at (-3, 0.5) {};
        \node [style=none] (14) at (3.5, 0.5) {};
        \node [style=square box, minimum width=1 cm] (15) at (-3.5, 0) {$f$};
        \node [style=square box, minimum width=1 cm] (16) at (3.5, 0) {$f'$};
        \node [style=none] (17) at (-3.5, -0.5) {};
        \node [style=none] (18) at (3, -0.5) {};
        \node [style=none] (19) at (4, -0.5) {};
        \node [style=none] (20) at (1.5, -0.75) {};
        \node [style=none] (21) at (3, -0.75) {};
        \node [style=none, font=\footnotesize] (22) at (-3.75, -1.25) {$B$};
        \node [style=none, font=\footnotesize] (23) at (4.25, -1.25) {$B$};
        \node [style=none] (24) at (-3.5, -1.5) {};
        \node [style=none] (25) at (-1.5, -1.5) {};
        \node [style=dot] (26) at (2.25, -1.5) {};
        \node [style=none] (27) at (4, -1.5) {};
        \node [style=none] (28) at (5, -1.5) {};
    \end{pgfonlayer}
    \begin{pgfonlayer}{edgelayer}
        \draw (2.center) to (12.center);
        \draw (8.center) to (13.center);
        \draw[out=0, in=90] (0) to (9.center);
        \draw (17.center) to (24.center);
        \draw[out=180, in=90] (1) to (10.center);
        \draw (11.center) to (28.center);
        \draw (9.center) to (25.center);
        \draw (10.center) to (14.center);
        \draw (19.center) to (27.center);
        \draw[out=180, in=90] (0) to (8.center);
        \draw[out=0, in=90] (1) to (11.center);
        \draw (21.center) to (18.center);
        \draw[out=0, in=-90] (26) to (21.center);
        \draw (20.center) to (3.center);
        \draw[out=180, in=-90] (26) to (20.center);
    \end{pgfonlayer}
\end{tikzpicture}
\]

The axioms of a GHZ/W-pair suffice to prove the following lemma for Hilbert spaces.
\begin{lemma}[\cite{CK}]\label{lem:dot-lolli-2d}
For a GHZ/W-pair on $H$ in ${\bf FHilb}$ with $\dim(H) \geq 2$, the points $\unit$ and $\tickunit$ span a 2-dimensional space; hence for $H=\mathbb{C}^2$ the points $\unit$ and $\tickunit$ form a basis.
\end{lemma}

Motivated by this fact, we assume the that $\{ \unit, \tickunit \}$ forms a plugging set for $Q$. More explicitly:
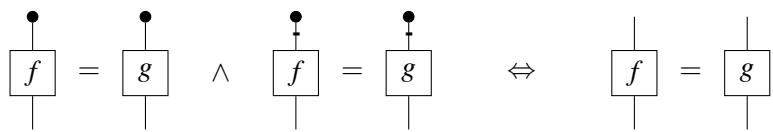
\[
\begin{tikzpicture}
    \begin{pgfonlayer}{nodelayer}
        \node [style=dot] (0) at (-6, 0.75) {};
        \node [style=dot] (1) at (-4.5, 0.75) {};
        \node [style=dot] (2) at (-2.5, 0.75) {};
        \node [style=dot] (3) at (-1, 0.75) {};
        \node [style=none] (4) at (2, 0.75) {};
        \node [style=none] (5) at (3.5, 0.75) {};
        \node [style=square box] (6) at (-6, 0) {$f$};
        \node [style=none] (7) at (-5.25, 0) {$=$};
        \node [style=square box] (8) at (-4.5, 0) {$g$};
        \node [style=none] (9) at (-3.5, 0) {$\wedge$};
        \node [style=square box] (10) at (-2.5, 0) {$f$};
        \node [style=none] (11) at (-1.75, 0) {$=$};
        \node [style=square box] (12) at (-1, 0) {$g$};
        \node [style=none] (13) at (0.5, 0) {$\Leftrightarrow$};
        \node [style=square box] (14) at (2, 0) {$f$};
        \node [style=none] (15) at (2.75, 0) {$=$};
        \node [style=square box] (16) at (3.5, 0) {$g$};
        \node [style=none] (17) at (-6, -0.75) {};
        \node [style=none] (18) at (-4.5, -0.75) {};
        \node [style=none] (19) at (-2.5, -0.75) {};
        \node [style=none] (20) at (-1, -0.75) {};
        \node [style=none] (21) at (2, -0.75) {};
        \node [style=none] (22) at (3.5, -0.75) {};
    \end{pgfonlayer}
    \begin{pgfonlayer}{edgelayer}
        \draw (10) to (19.center);
        \draw (0) to (6);
        \draw (2) to node[tick]{-} (10);
        \draw (16) to (22.center);
        \draw (8) to (18.center);
        \draw (1) to (8);
        \draw (6) to (17.center);
        \draw (12) to (20.center);
        \draw (14) to (21.center);
        \draw (4.center) to (14);
        \draw (5.center) to (16);
        \draw (3) to node[tick]{-} (12);
    \end{pgfonlayer}
\end{tikzpicture}
\]

\section{Arithmetic from a GHZ/W-pair}\label{sec:arithmetic}

Given a GHZ/W-pair, we can extract an arithmetic system.  First, we establish some preliminary results.

\subsection{Properties of GHZ-phases}

Below, all phases are GHZ-phases. When relying on plugging, we have the following:

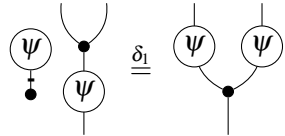
\newcommand{\psidot}{
\raisebox{5mm}{\begin{tikzpicture}[scale=0.6]
    \begin{pgfonlayer}{nodelayer}
        \node [style=white dot] (0) at (0, 0.5) {$\psi$};
        \node [style=dot] (1) at (0, -0.5) {};
    \end{pgfonlayer}
    \begin{pgfonlayer}{edgelayer}
        \draw (0) to (1);
    \end{pgfonlayer}
\end{tikzpicture}}}

\newcommand{\psitickdot}{
\raisebox{5mm}{\begin{tikzpicture}[scale=0.6]
    \begin{pgfonlayer}{nodelayer}
        \node [style=white dot] (0) at (0, 0.5) {$\psi$};
        \node [style=dot] (1) at (0, -0.5) {};
    \end{pgfonlayer}
    \begin{pgfonlayer}{edgelayer}
        \draw (0) to node[tick]{-} (1.center);
    \end{pgfonlayer}
\end{tikzpicture}}}

\newcommand{\smallpsidot}{
\raisebox{-2.5mm}{\begin{tikzpicture}[scale=0.5]
    \begin{pgfonlayer}{nodelayer}
        \node [style=white dot] (0) at (0, 0.5) {\footnotesize$\!\psi$};
        \node [style=dot] (1) at (0, -0.5) {};
    \end{pgfonlayer}
    \begin{pgfonlayer}{edgelayer}
        \draw (0) to (1);
    \end{pgfonlayer}
\end{tikzpicture}}}

\newcommand{\smallpsitickdot}{
\raisebox{-2.5mm}{\begin{tikzpicture}[scale=0.5]
    \begin{pgfonlayer}{nodelayer}
        \node [style=white dot] (0) at (0, 0.5) {\footnotesize$\!\psi$};
        \node [style=dot] (1) at (0, -0.5) {};
    \end{pgfonlayer}
    \begin{pgfonlayer}{edgelayer}
        \draw (0) to node[tick]{-} (1.center);
    \end{pgfonlayer}
\end{tikzpicture}}}

\begin{theorem}\label{thmdelta_1}
\[
\raisebox{-8mm}{
\psitickdot \
}
\]
\end{theorem}

\begin{proof}
Plugging $\unit$ to one input:
    \[
    \raisebox{-8mm}{
    \psitickdot \
    %
}
    \]
Plugging $\tickunit$ to one input of both sides:
    \[
    \raisebox{-8mm}{
    \psitickdot \
    %
}
\]
\end{proof}

Note that for  $\whitemult$-phases we have:
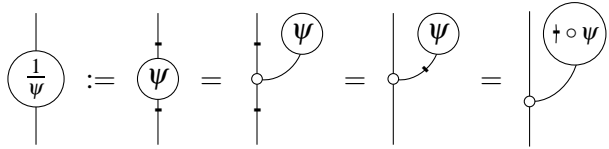
\[
\raisebox{-8mm}{%
}
\]
The particular choice of notation $\frac{1}{\psi}$ is justified below, and will play a key role in this paper.

\begin{theorem}\label{thmdelta_3}
\[
\raisebox{-8mm}{%
}
\]
\end{theorem}

\begin{proof}
Plugging $\unit$ into the input:
    \[
    \raisebox{-8mm}{%
}
    \]
Plugging $\tickunit$ into the input:
    \[
    \raisebox{-8mm}{%
}
    \]
\end{proof}

In a setting like ${\bf FHilb}_p$, where we ignore cancelable scalar multipliers, and provided
that the scalars $\smallpsidot$ and $\smallpsitickdot$ in the equation are cancelable, eqs.~($\delta_1$, $\delta_2$, $\delta_3$) simplify to:
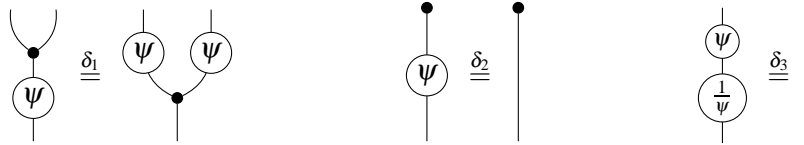
\[
\raisebox{-8mm}{%
}
\]
Examples of phases for which some of these simplified equations fail to hold are:
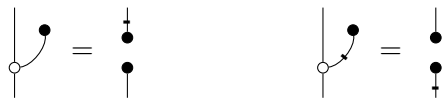
\[
%
\]

\subsection{Natural Number Arithmetic}\label{sec:natural-numbers}

Assume that we are given a GHZ/W-pair. In particular, we have two internal commutative monoids $(\mult, \unit)$ and $(\whitemult, \whiteunit)$.  We will now consider the induced commutative monoids on elements:
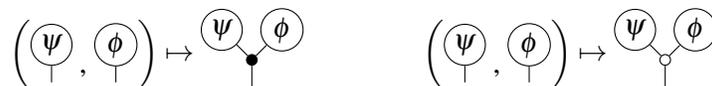
\[
\left(
\raisebox{-3mm}{%
}
\]
 We will call $\mult$ applied to elements  \em addition \em and $\whitemult$ applied to  elements
\em multiplication\em, for reasons that will become apparent shortly.  Similarly, we call $\unit$
the \em unit for addition \em and $\whiteunit$ the \em unit for multiplication\em. By Theorem
\ref{thmdelta_1} we have a distributivity law, up to a scalar, partially explaining our choices of
the names addition and multiplication for the monoids:

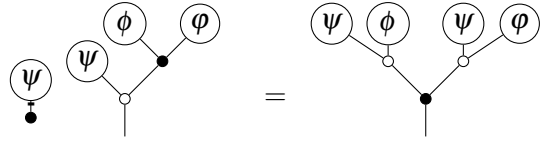
\begin{corollary}\label{cor-dist}
\[
%
\]
\end{corollary}

Moreover, we can use these to do concrete arithmetic on the natural numbers.  We start by defining an encoding for the natural numbers:
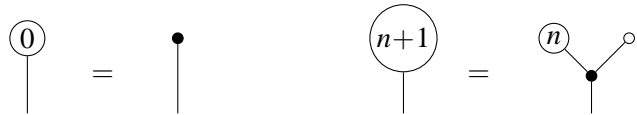
\[
%
\]
From hence forth, we shall assume we are working in a category with no non-trivial invertible (i.e.
non-zero) scalars, such as ${\bf FHilb}_p$. Thus, we shall drop any invertible scalars.
Furthermore, we shall assume scalar (i.) is invertible for all $n$ and scalar (ii.) is invertible
for all $n \neq 0$:
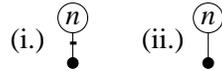
\[
%
\]

That $\mult$ is the normal addition operation for these numbers follows immediately from their
definition and associativity of $\mult$.  We can also show that $\whitemult$ is the normal
multiplication operation (noting first that the encoding of $1$ is $\whiteunit$, and hence the unit
of $\whitemult$):

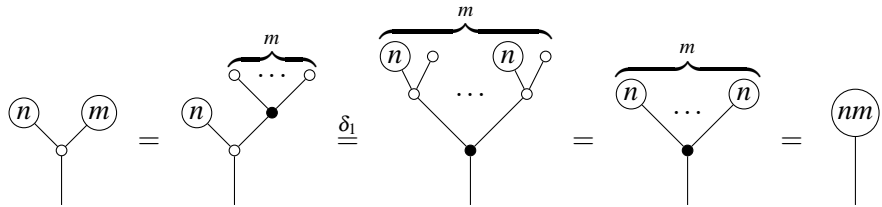
\[
\raisebox{-8mm}{%
}
\]

The distributivity law stated earlier now translates into the normal distributivity of multiplication over addition in the natural numbers, up to a scalar.

\subsection{Multiplicative inverses}\label{sec:rationals}

By Theorem \ref{thmdelta_3} we have that $\tick$ is also an inverse for $\whitemult$, up to a scalar:

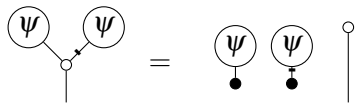
\begin{corollary}\label{cor-inverse}
\[
%
\]
\end{corollary}
Hence, we have an encoding for the multiplicative inverses of the natural numbers:
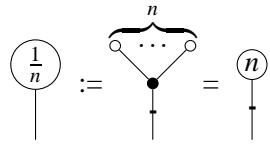
\[
%
\]
Thoughtout this paper, we shall assume any natural number occurring in the denominator is not equal
to $0$. This allows us to encode positive fractions in the following form:
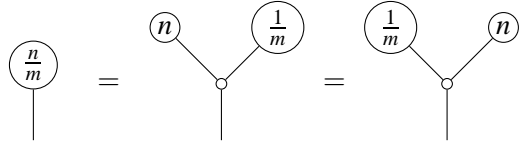
\[
%
\]
where the second equality follows from associativity and commutativity of the GHZ-structure.

\begin{remark}
It should be noted that the construction of this encoding depends on the choice of numerator and denominator,
and not just on the rational number being represented.  Therefore, we should demonstrate that the actual point
depends only on the number represented.  We start by noting that, by corollary \ref{cor-inverse},
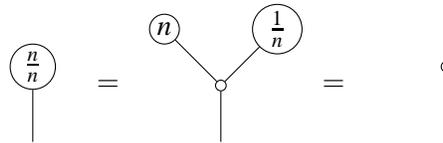
\[
%
\]
where we have ignored the scalar in corollary \ref{cor-inverse}, since it is cancellable
for the points that we have used to encode the natural numbers.

We know that if $\frac{n}{m} = \frac{n'}{m'}$, there
is a $p$ such that $n = n'p$ and $m = m'p$, so it follows easily that
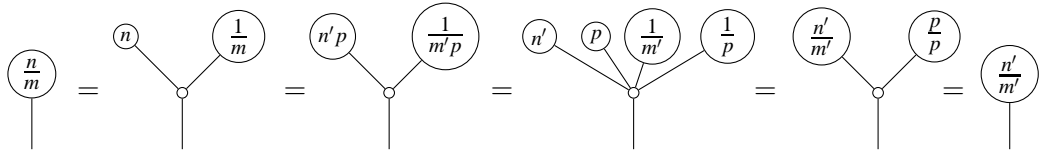
\[
%
\]
\end{remark}

\begin{example}\label{ex:invfac}
All the usual properties of fractions follow in a straightforward manner from the axioms of rational arithmetic that we have proved.  For example,
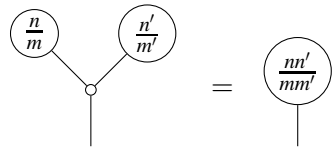
\[
%
\]
is immediate from the associativity commutativity of the GHZ-structure, and we also have:
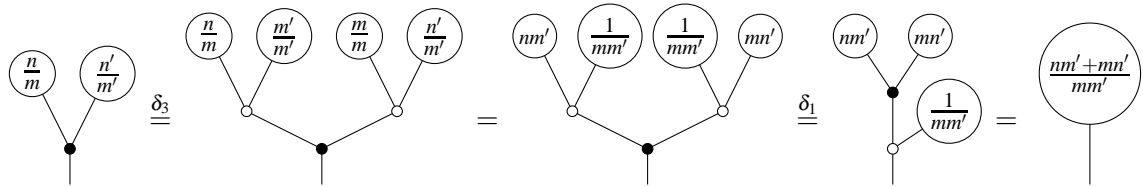
\[
\raisebox{-8mm}{%
}
\]
\end{example}
where we used distributivity and ignored the scalars; this is fine as the
scalars in corollary \ref{cor-dist} are cancellable for all the points
we are considering.

\section{Additive inverses}\label{sec:additive-inverses}

The last thing we need to actually produce a field of fractions is additive inverses.
Suppose we have an involutive operation $\cross$ that is a phase for $\whitemult$:
\[
%
\]
Suppose further that $\{ \whiteunit,\ \whitecrossunit \}$ forms a plugging set:
\[
%
\]
Then this will act as an additive inverse operation.

\begin{lemma}\label{lem:pauli-z-homom}
\[
%
}
\]
\end{lemma}

\begin{proof}
Follows immediately from Lemma \ref{lem:pauli-z-homom} and the fact that $\cross$ is an involution.
\end{proof}

\begin{theorem}\label{thm:add-inv}
The cross is additive inverse:
\[
\raisebox{-8mm}{%
}
\]
\end{theorem}

\begin{proof}
Proof by plugging: Recall that
%
}
        \]
\end{enumerate}
\end{proof}

In the case of ${\bf FHilb}$, this operation is the Pauli Z gate multiplied by $-1$:
\[
\cross
\ \ = \ \
\left(\begin{array}{cc}
-1 & 0 \\
0 & 1
\end{array}\right)
\]
In this case, we have
\[
%
\ \ = \ \
-\sqrt{2} \ket{-}
\ \ = \ \
\ket{1} - \ket{0}
\]
We can now naturally define:
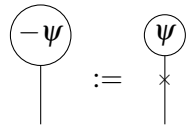
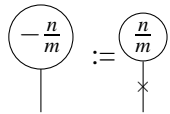
\[
%
\]
and, in particular,
\[
%
\]

Thus, we have reconstructed all of the axioms for the field of rational numbers. All of the expected identities involving $\cross$ then follow from the field axioms.

\section{!-boxes and automation}\label{sec:automation}

Graph rewriting is a computation process in which graphs are transformed by various \emph{rewrite
rules}. A rewrite rule can be thought of as a `directed graphical equation'. For example, the
``specialness'' equation from section \ref{sec:ghz-w} could be expressed as a graph rewrite rule:
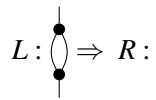
\[
L: %
\]

A graph rewrite rule $L \Rightarrow R$ can be applied to a graph $G$ by identifying a
\emph{matching}, that is, a monomorphism $m : L \rightarrow G$. The image of $L$ under $m$ is then
removed and replaced by $R$. This process is called \emph{double pushout (DPO) graph rewriting}. A
detailed description of how DPO graph rewriting can be performed on the graphs described in this
paper is available in \cite{DK}.

It is also useful to talk not only about `concrete' graph rewrite rules, but also \emph{pattern
graph} rewrite rules, which can be used to express an infinite set of rewrite rules. We form
pattern graphs using \emph{!-boxes} (or `bang-boxes'). These boxes identify portions of the graph
that can be replicated any number of times. More precisely, the set of concrete graphs represented
by a pattern graph is the set of all graphs (containing no !-boxes) that can be obtained by
performing any sequence of these four operations:
\begin{itemize}
    \item \texttt{COPY}: copy a !-box and its incident edges
    \item \texttt{MERGE}: merge two !-boxes
    \item \texttt{DROP}: remove a !-box, leaving its contents
    \item \texttt{KILL}: remove a !-box and its contents
\end{itemize}
For example, the following pattern graph represents the encoding of a natural number given in section \ref{sec:natural-numbers}:
\[\left\llbracket\ \
\begin{tikzpicture}[dotpic]
    \begin{pgfonlayer}{nodelayer}
        \node [style=box vertex] (0) at (0, 0.75) {};
        \node [style=white dot] (1) at (0, 0.75) {};
        \node [style=dot] (2) at (0, 0) {};
        \node [style=none] (3) at (0, -0.75) {};
    \end{pgfonlayer}
    \begin{pgfonlayer}{edgelayer}
        \draw (2) to (3.center);
        \draw (1) to (2);
    \end{pgfonlayer}
\end{tikzpicture}\ \ \right\rrbracket =
\left\{\
\begin{tikzpicture}[dotpic]
    \begin{pgfonlayer}{nodelayer}
        \node [style=dot] (0) at (0, 0) {};
        \node [style=none] (1) at (0, -0.75) {};
    \end{pgfonlayer}
    \begin{pgfonlayer}{edgelayer}
        \draw (0) to (1.center);
    \end{pgfonlayer}
\end{tikzpicture}\ ,\ \
\begin{tikzpicture}[dotpic]
    \begin{pgfonlayer}{nodelayer}
        \node [style=white dot] (0) at (0, 0.75) {};
        \node [style=dot] (1) at (0, 0) {};
        \node [style=none] (2) at (0, -0.75) {};
    \end{pgfonlayer}
    \begin{pgfonlayer}{edgelayer}
        \draw (1) to (2.center);
        \draw (0) to (1);
    \end{pgfonlayer}
\end{tikzpicture}\ ,\ \
\begin{tikzpicture}[dotpic]
    \begin{pgfonlayer}{nodelayer}
        \node [style=white dot] (0) at (-0.5, 0.75) {};
        \node [style=white dot] (1) at (0.5, 0.75) {};
        \node [style=dot] (2) at (0, 0) {};
        \node [style=none] (3) at (0, -0.75) {};
    \end{pgfonlayer}
    \begin{pgfonlayer}{edgelayer}
        \draw (1) to (2);
        \draw (2) to (3.center);
        \draw (0) to (2);
    \end{pgfonlayer}
\end{tikzpicture},\ \
\begin{tikzpicture}[dotpic]
    \begin{pgfonlayer}{nodelayer}
        \node [style=white dot] (0) at (-0.75, 0.75) {};
        \node [style=white dot] (1) at (0, 0.75) {};
        \node [style=white dot] (2) at (0.75, 0.75) {};
        \node [style=dot] (3) at (0, 0) {};
        \node [style=none] (4) at (0, -0.75) {};
    \end{pgfonlayer}
    \begin{pgfonlayer}{edgelayer}
        \draw (2) to (3);
        \draw (0) to (3);
        \draw (3) to (4.center);
        \draw (1) to (3);
    \end{pgfonlayer}
\end{tikzpicture},\ \
\ldots\
\right\} :=
\left\{
\point{0},\ \point{1},\ \point{2},\ \point{3},\ \ldots
\right\}\]
Note how not only the vertices are duplicated, but also all of the edges connected to those vertices.

A pattern graph rewrite rule is a pair of pattern graphs with the same inputs and outputs.
Furthermore, there is a bijection between the !-boxes occurring on the LHS and the RHS. When one of
the four operations is performed to a !-box on the LHS, the same is performed to its corresponding
!-box on the RHS.

We can rewrite (the natural numbers versions of) equations $\delta_1$, $\delta_2$, and $\delta_3$ as pattern graph rewrite rules.
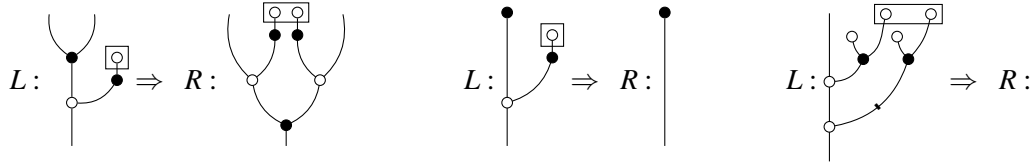
\[
L:\raisebox{-8mm}{\begin{tikzpicture}[scale=0.6]
    \begin{pgfonlayer}{nodelayer}
        \node [style=none] (0) at (-0.5, 1.5) {};
        \node [style=none] (1) at (0.5, 1.5) {};
        \node [style=dot] (2) at (0, 0.5) {};
        \node [style=white dot] (3) at (0, -0.5) {};
        \node [style=none] (4) at (0, -1.5) {};
        \node [style=white dot] (5) at (1, 0.5) {};
        \node [style=box vertex,inner sep=1.5mm] (6) at (1, 0.5) {};
        \node [style=dot] (7) at (1, 0) {};
    \end{pgfonlayer}
    \begin{pgfonlayer}{edgelayer}
        \draw [bend right=30] (0) to (2);
        \draw [bend right=30] (2) to (1);
        \draw (2) to (3);
        \draw (3) to (4);
        \draw [bend right=30] (3) to (7);
        \draw (7) to (5);
    \end{pgfonlayer}
\end{tikzpicture}}
 \Rightarrow\ \
R: \raisebox{-8mm}{\begin{tikzpicture}[scale=0.6]
    \begin{pgfonlayer}{nodelayer}
        \node [style=none] (0) at (-1.25, 1.5) {};
        \node [style=white dot] (1) at (-0.25, 1.5) {};
        \node [style=white dot] (2) at (0.25, 1.5) {};
        \node [style=none] (3) at (1.25, 1.5) {};
        \node [style=dot] (4) at (-0.25, 1) {};
        \node [style=dot] (5) at (0.25, 1) {};
        \node [style=white dot] (6) at (-0.75, 0) {};
        \node [style=white dot] (7) at (0.75, 0) {};
        \node [style=dot] (8) at (0, -1) {};
        \node [style=none] (9) at (0, -1.5) {};
        \node [style=box vertex,minimum height=2mm,minimum width=6mm] (10) at (0, 1.5) {};
    \end{pgfonlayer}
    \begin{pgfonlayer}{edgelayer}
        \draw [bend right=30] (0) to (6.center);
        \draw [bend right=30] (6.center) to (4.center);
        \draw (4.center) to (1);
        \draw [bend right=30] (6.center) to (8.center);
        \draw [bend right=30] (8.center) to (7.center);
        \draw [bend right=30] (5.center) to (7.center);
        \draw [bend right=30] (7.center) to (3);
        \draw (5.center) to (2);
        \draw (8.center) to (9);
    \end{pgfonlayer}
\end{tikzpicture}}
\qquad \qquad
L:\raisebox{-8mm}{\begin{tikzpicture}[scale=0.6]
    \begin{pgfonlayer}{nodelayer}
        \node [style=dot] (0) at (0, 1.5) {};
        \node [style=none] (1) at (0, -1.5) {};
        \node [style=white dot]  (2) at (0, -0.5) {};
        \node [style=white dot] (3) at (1, 1) {};
        \node [style=dot] (4) at (1, 0.5) {};
        \node [style=box vertex,inner sep=1.5mm] (5) at (1, 1) {};
    \end{pgfonlayer}
    \begin{pgfonlayer}{edgelayer}
        \draw (0) to (2.center);
        \draw (2.center) to (1);
        \draw [bend right=30] (2.center) to (4.center);
        \draw (4.center) to (3.center);
    \end{pgfonlayer}
\end{tikzpicture}}
 \Rightarrow\ \
R: \raisebox{-8mm}{\begin{tikzpicture}[scale=0.6]
    \begin{pgfonlayer}{nodelayer}
        \node [style=dot] (0) at (0, 1.5) {};
        \node [style=none] (1) at (0, -1.5) {};
    \end{pgfonlayer}
    \begin{pgfonlayer}{edgelayer}
        \draw (0)--(1);
    \end{pgfonlayer}
\end{tikzpicture}}
\qquad \qquad
L:\begin{tikzpicture}[dotpic]
    \begin{pgfonlayer}{nodelayer}
        \node [style=none] (0) at (0, 1.5) {};
        \node [style=white dot] (1) at (1.25, 1.5) {};
        \node [style=box vertex, minimum width=9 mm, minimum height=2 mm] (2) at (1.75, 1.5) {};
        \node [style=white dot] (3) at (2.25, 1.5) {};
        \node [style=white dot] (4) at (0.5, 1) {};
        \node [style=white dot] (5) at (1.5, 1) {};
        \node [style=dot] (6) at (0.75, 0.5) {};
        \node [style=dot] (7) at (1.75, 0.5) {};
        \node [style=white dot] (8) at (0, 0) {};
        \node [style=white dot] (9) at (0, -1) {};
        \node [style=none] (10) at (0, -1.75) {};
    \end{pgfonlayer}
    \begin{pgfonlayer}{edgelayer}
        \draw (9) to (10.center);
        \draw[bend right=15] (4) to (6);
        \draw (8) to (9);
        \draw (0.center) to (8);
        \draw[out=30, in=258] (6) to (1);
        \draw[bend right=15] (5) to (7);
        \draw[bend right] (9) to node[tick]{-} (7);
        \draw[bend right] (8) to (6);
        \draw[bend right=15] (7) to (3);
    \end{pgfonlayer}
\end{tikzpicture}
 \Rightarrow\ \
R: \raisebox{-9mm}{\begin{tikzpicture}[scale=0.6]
    \begin{pgfonlayer}{nodelayer}
        \node [style=none] (0) at (0, 1.5) {};
        \node [style=none] (1) at (0, -1.5) {};
    \end{pgfonlayer}
    \begin{pgfonlayer}{edgelayer}
        \draw (0)--(1);
    \end{pgfonlayer}
\end{tikzpicture}}
\]
These equations only apply to encodings of the natural numbers, not for arbitrary inputs $\psi$.
However, we showed in sections \ref{sec:natural-numbers}, \ref{sec:rationals}, and
\ref{sec:additive-inverses} that even those weaker equations suffice to recover the usual
identities for fraction arithmetic. Also note that the extra white vertices in $\delta_3'$
eliminate the case of $\frac{0}{0} \neq 1$.

So why bother expressing graphical identities as graph rewrite rules? Graph rewriting can be
automated! {\tt quantomatic}~\cite{quanto} is a automatic graph rewriting tool developed by two of
the authors. It is specifically designed to work with the kinds of diagram described in this paper
and to perform pattern graph rewriting. It remains to be seen what new insights can be obtained by
adding the new graphical identities derived in this paper to {\tt quantomatic}.

\section{Closing Remarks}

In previous work two of the authors showed that the main difference between the GHZ state and the W state, or more precisely, the induced GHZ structure and W structure, boils down to the value of the loop map of these CFAs:
\[
\frac{\mbox{GHZ}}{\mbox{W}}\ \ =\ \ \frac{\begin{tikzpicture}[dotpic,yshift=5mm]
       \node [dot] (a) at (0,0) {};
       \node [dot] (b) at (0,-1) {};
       \draw [bend left] (a) to (b);
       \draw [bend right] (a) to (b);
       \draw (0,0.5) to (a) (b) to (0,-1.5);
\end{tikzpicture} = \
\begin{tikzpicture}[dotpic]
       \draw (0,1) -- (0,-1);
\end{tikzpicture}
}{
\circl\ \begin{tikzpicture}[dotpic]
               \node [dot] (a) at (0,0.5) {};
               \node [dot] (b) at (0,-0.5) {};
               \draw [bend left] (a) to (b);
               \draw [bend right] (a) to (b);
               \draw (0,1) to (a) (b) to (0,-1);
       \end{tikzpicture}\ \  =
       \begin{tikzpicture}[dotpic]
               \node [dot] (a) at (0,0.7) {};
               \node [dot] (b) at (0,-0.7) {};
               \draw (0,1.2) to (a) (b) to (0,-1.2);
               \draw (a) to [downloop] ();
               \draw (b) to [uploop] ();
       \end{tikzpicture}
}
\]
In this paper, by focussing on the interaction of these two structures, we were able to establish a connection with the operations of basic arithmetic:
\[
\frac{\mbox{W}}{\mbox{GHZ}}=\frac{+}{\times}
\]
More specifically, the diagrammatic language of these structures was sufficient to encode the positive rational numbers (and, with a minor extension, the whole field of rational numbers).

In the process of highlighting this encoding, we identified a surprising fact. The distributive law
governing the interaction of addition and multiplication in arithmetic also captures the
interaction of the GHZ-structure and W-structure. Future work includes exploiting this interaction
in the study of multipartite quantum entanglement, which brings us back to the initial motivation
for crafting a compositional framework to reason about multipartite states.

\bibliography{rationals}

\end{document}